\documentclass[12pt,a4paper,final]{iopart}
\usepackage{iopams}  
\usepackage{graphicx}
\usepackage[breaklinks=true,colorlinks=true,linkcolor=blue,urlcolor=blue,citecolor=blue]{hyperref}
\usepackage{setstack}

\usepackage{amsfonts}
\usepackage{amsthm}
\usepackage{amssymb}
\usepackage{url}
\usepackage{hyperref}
\usepackage{color}

\input{Qcircuit}

\newcommand{\braket}[2]{\left \langle #1 | #2 \right \rangle}
\newcommand{\norm}[1]{\left|\left|#1\right|\right|}

\usepackage{color}

\newtheorem{theorem}{Theorem}

\newtheorem{lemma}{Lemma}

\newtheorem{claim}{Claim}

\begin{document}
\title{Self-testing in parallel}

\author{Matthew McKague$^{1,2}$}
\address{$^1$Department of Computer Science, University of Otago}
\address{$^2$Dodd-Walls Centre for Photonic and Quantum Technologies}
\ead{mckaguem@cs.otago.ac.nz}

\begin{abstract}
Self-testing allows us to determine, through classical interaction only, whether some players in a non-local game share particular quantum states.  Most work on self-testing has concentrated on developing tests for small states like one pair of maximally entangled qubits, or on tests where there is a separate player for each qubit, as in a graph state.  Here we consider the case of testing many maximally entangled pairs of qubits shared between two players.  Previously such a test was shown where testing is sequential, i.e.,\ one pair is tested at a time.  Here we consider the parallel case where all pairs are tested simultaneously, giving considerably more power to dishonest players.  We derive sufficient conditions for a self-test for many maximally entangled pairs of qubits shared between two players and also two constructions for self-tests where all pairs are tested simultaneously.

\end{abstract}

\section{Introduction}

A non-local game is a scenario where two or more non-communicating players receive challenges or questions from a referee.  The players respond to the referee who subsequently announces whether they have won the game.  The possible questions and conditions for winning are publicly known.  We are interested in the case where the players have quantum capabilities and may share entanglement, but where communication with the referee is classical.  We will also concentrate on the case of two players, known as Alice and Bob.

Self-testing allows us to verify the functioning of quantum devices without reference to any trusted equipment.  Two or more players are queried with classical strings and their classical outputs are observed.  By checking these outputs we can, for some special cases, decide that the devices share a particular ideal state, and that they measure it according to some ideal operators (up to some equivalence).  In the language of non-local games, self-testing means that for some non-local games there essentially exists only one strategy that obtains the maximal probability of winning.  These results are also robust, allowing us to put error bounds on the state in the case of small amounts of noise in the devices.  Indeed, the robustness of these results is essential for applications, since we can never measure the probability of winning exactly.

A natural and desirable extension to any self-test would be the ability to repeat it, allowing us to self-test many copies of the same state.  This can be done quite naturally if the many copies are held in separate non-communicating devices.  But what if we cannot guarantee separation of a large number of states?  For example, we may wish to test a large number of pairs of maximally entangled qubits (each pair is known as one \emph{e-bit}) where one qubit from each pair is held by Alice and the other qubit in each pair is held by Bob.  Here we lose the very useful division of pairs into tensor products with each other.  Instead we have the considerably weaker division into only two subsystems.

\subsection{Parallel and sequential testing}

There are two natural ways of performing a test many times over.  The first is \emph{sequential} testing, where each test is completed before the next test begins.  For each test the referee will send the questions to the participants and wait for their responses before proceeding to the next test.  In \emph{parallel} testing, by contrast, all tests are in progress at the same time.  Here the referee will send all questions for all tests to the participants at the same time and the participants will likewise send their answers for all tests together.  

It is possible to make a sequential test into a parallel test by revealing all of the questions at once.  We can also turn a parallel test into a sequential test by revealing the questions a piece at a time.  In both cases, the parallel test allows for more strategies for the players because they have more information.  Indeed, a sequential strategy can always be played as a parallel strategy, but the opposite is not always true.  For us, this means that bounds that we place on strategies from parallel tests will also be valid for sequential tests.

In the strict sense, parallel testing would mean that we have many sub-tests, each of which is completely independent of the other sub-tests, i.e.,\ there would be fresh independent randomness for each test and all possible combinations of questions for sub-tests could be asked.  More generally, we might not have independent questions for each sub-test.  We will develop a strictly parallel test and also a test which is not strictly parallel that requires fewer questions.

\subsection{Previous work}
Self-testing was introduced by Mayers and Yao in \cite{Mayers:2004:Self-testing-qu} where a test now known as the Mayers-Yao test was developed,  with robustness bounds appearing in \cite{Magniez:2006:Self-testing-of}.  The CHSH test \cite{Clauser:1969:Proposed-Experi} is also known to be a self-test, a result which is implied in \cite{Popescu:1992:Which-states-vi} with initial robustness results in \cite{Bardyn:2009:Device-independ}. Robustness results for the Mayers-Yao test and the CHSH test appear in \cite{McKague:2012:Robust-self-tes}, which we use here.  We will also use techniques from \cite{McKague:2013:Interactive-proofs-for-BQP-via-self-tested-graph-states} which were originally developed for testing graph states.

The idea of using a self-test many times in parallel is used is \cite{Magniez:2006:Self-testing-of}, where the separate tests are assumed to be on separate subsystems.  More recently, Reichardt et al.\ \cite{Reichardt2013:Classicalcommandofquantumsystems} proved that sequential repetition of CHSH games can be used for self-testing.  Wu et al.\ \cite{Wu2015:Deviceindependentself} consider the case of two CHSH games played in parallel.

\subsection{Contributions}
We make three main contributions in this paper.  First we generalize the techniques used in \cite{McKague:2013:Interactive-proofs-for-BQP-via-self-tested-graph-states} Theorem~3, allowing us to consider the case of only two players for testing multi-qubit states.   We then use this to derive sufficient conditions for self-testing many e-bits shared between two players.  The two-player multi-qubit result has independent value for self-testing, and is used in \cite{Wu2015:Deviceindependentself} to prove that the magic square game is a self-test.

The second contribution is to develop a test -- based on the Mayers-Yao test for a single e-bit -- that self-tests many e-bits.  Interestingly, the test requires only a logarithmic number of measurements (in the number of e-bits tested), and has a robustness bound that scales linearly.

Our final contribution is to develop a self-test for many e-bits which is strictly parallel.  The basic test is new, and can be seen as an extension of CHSH.  This test has an exponential number of measurement settings, and the robustness also scales exponentially in the number of sub-tests.  We also phrase this self-test as a non-local game, defining winning conditions for one round of the test, and give a robustness bound on how far away the strategy is from the ideal in terms of the winning probability.

\section{Technical preliminaries}
We define $1_k$ to be the $n$-bit string which is 1 in the $k$-th position and 0 everywhere else.  For $x$ an $n$-bit string, let $|x|$ be the number of 1's in $x$ (the Hamming weight).  Further, when $n$ is even, define $x_a$ to be the $n$-bit string that agrees with $x$ for the first $\frac{n}{2}$ bits and is zero elsewhere. The $n$-bit string $x_b$ agrees with $x$ for the last $\frac{n}{2}$ bits and is zero elsewhere.  Later, we will divide $x$ between Alice and Bob so $x_a$ represents $x$ on Alice's side, and $x_b$ is for Bob's side.  The matrix $R$ exchanges the first and second haves of a bit string.  Since $R$ simply applies a fixed permutation on the entries of a string, it preserves the dot product, so that $Rx \cdot Ry = x \cdot y$.  Also, $R^2 = I$ so $Rx \cdot y = x \cdot Ry$.

We will be dealing with sums over all bit strings, for which the following lemma will be invaluable.
\begin{lemma} For $s,t$ ranging over $\{0,1\}^n$
\label{lemma:stringsums}
\begin{eqnarray}
\frac{1}{2^n}\sum_s s \cdot t & = & \frac{|t|}{2} \\
\frac{1}{2^{2n}}  \sum_{s,t} s \cdot t & = & \frac{n}{4} \\
\frac{1}{2^n}\sum_s (-1)^{s \cdot t} & = & \delta_{t, 0}
\end{eqnarray}
\end{lemma}

The proofs of these equations are straightforward and left to the reader.  Next we need to know something about how $R$ behaves.  This will later be used to simplify the phases in a graph state when $R$ is the adjacency matrix.  This corresponds to Lemma~2 in \cite{McKague:2013:Interactive-proofs-for-BQP-via-self-tested-graph-states}.

\begin{lemma}
\label{lemma:Rdecompose}
Let $s$ and $u$ be $n$-bit strings with $n$ even.  Then
\begin{equation}
(R(s \oplus u) \cdot s)
\oplus 
(R(s \oplus u)_a \cdot (s \oplus u)_b )
=
(R s_b \cdot s_a ) 
\oplus
(R u_a \cdot u_b )
\end{equation}.
\end{lemma}

\begin{proof}
Exploiting linearity and the fact that $x_a \cdot y_b = 0$ for any $x,y$ we find
\begin{eqnarray}
\nonumber
(R(s \oplus u) \cdot s)
\oplus 
(R(s \oplus u)_a \cdot (s \oplus u)_b )
= \\
\qquad 
\qquad
\qquad
\qquad
(R(s \oplus u)_b \cdot s_a )
\oplus 
(R(s \oplus u)_a \cdot s_b )
\oplus \\
\nonumber
\qquad
\qquad
\qquad
\qquad
\qquad
\qquad
\qquad
(R(s \oplus u)_a \cdot s_b)
\oplus
(R(s \oplus u)_a \cdot u_b).
\end{eqnarray}
The centre two terms on the right side now cancel.  Expanding once again we get
\begin{eqnarray}
\nonumber
(R(s \oplus u) \cdot s)
\oplus 
(R(s \oplus u)_a \cdot (s \oplus u)_b )
= \\
\qquad 
\qquad
\qquad
\qquad
(R s_b \cdot s_a ) 
\oplus 
(R u_b \cdot s_a )
\oplus 
(R s_a \cdot u_b )
\oplus
(R u_a \cdot u_b )
\end{eqnarray}
Since $R$ preserves the dot product the two middle terms cancel, leaving us with the desired result.
\end{proof}

We will need a way to translate between inner products and the 2-norm.

\begin{lemma}
\label{lemma:innerproducttonorm}
Let $\ket{\psi_1}$ and $\ket{\psi_2}$ be normalized states.  If  $|\braket{\psi_1}{\psi_2}| \geq 1 - \epsilon$ for $\epsilon \geq 0$ then
        \begin{equation}
            \norm{
                \ket{\psi_1} - \ket{\psi_2}
            }_2 \leq \sqrt{2\epsilon}
        \end{equation}
\end{lemma}
The proof follows directly from the definition of $\norm{\cdot}_2$.  From now on, unless otherwise specified, $\norm{\cdot} = \norm{\cdot}_2$.

It is conventional to define a Pauli operator raised to a bit string by $P^t = \bigotimes_{k=1}^n P^{t_k}$.  We will adopt a generalization of this notation.  If we have operators $X_1 \dots X_n$ then for a bit string $t$ we define
\begin{equation}
X^t := \prod_{k=1}^n X_k^{t_k}.
\end{equation}
The order of the product is important since $X_j$ may not commute with some other $X_k$.  Hence we will make the convention that the index increases from left to right.  When applied to a state this ordering means that the operators are applied to the state in order of decreasing index.  With suitable modifications to the proof any other ordering will also work so long as it is kept consistent.  

\section{Testing with two players}

In this section we develop the infrastructure necessary to test many e-bits using just two players.  We divide this into two main parts.  The first part, Lemma~\ref{lemma:graphstateselftestconditions} -- which is a generalization of \cite{McKague:2013:Interactive-proofs-for-BQP-via-self-tested-graph-states} Theorem~3 -- also applies to graph states and makes no assumptions about any type of tensor product structure or commutation between any operators.  The change compared to \cite{McKague:2013:Interactive-proofs-for-BQP-via-self-tested-graph-states} is to provide a slightly more streamlined proof with a better bound, and slightly different isometry that does not rely on commutation relations between subsystems.

The second part uses Lemma~\ref{lemma:graphstateselftestconditions} to derive sufficient conditions for testing many e-bits.  The two necessary bounds are derived in Lemma~\ref{lemma:acandxz} and are analogous to the bounds given in \cite{McKague:2013:Interactive-proofs-for-BQP-via-self-tested-graph-states} Corollary~1 and Lemma~4.  There, however, a rich tensor-product structure was imposed by the division into many players, while here we have the much weaker structure provided by only two players.  The challenge, then, is to use this weaker structure to obtain much the same results, albeit with weaker bounds.  In Lemma~\ref{lemma:paralleltestsufficientconditions} we use these bounds along with Lemma~\ref{lemma:graphstateselftestconditions} to derive the sufficient conditions for testing many e-bits.

\subsection{Sufficient conditions for self-testing graph states}

\begin{lemma}
\label{lemma:graphstateselftestconditions}
Given an $n \times n$ $(0,1)$-matrix $\mathbf{A}$ and a function $P$ such that for all $s,t \in \{0,1\}^n$
\begin{equation}
\label{eq:PR9}
P(s) + P(t)
=
P(s \oplus t) + s \cdot \mathbf{A}(s \oplus t) \pmod{2}
\end{equation}
let $\ket{\psi}$ be the $n$-qubit state
\begin{equation}
\ket{\psi} = 
\frac{1}{\sqrt{2^n}}
\sum_u (-1)^{P(u)} \ket{u}.
\end{equation}
Further suppose that $\ket{\psi^{\prime}} \in \mathcal{H}$ is a normalized state, $\{X^{\prime}_j\}_{j=1}^n$, $\{Z^{\prime}_j\}_{j=1}^n$ are unitary, Hermitian operators on $\mathcal{H}$, and $\epsilon_{ac}(s,t) \geq 0 $ and $\epsilon_{xz}(s) \geq 0$ are functions such that for any $s,t \in (0,1)^{n}$
\begin{equation}
\label{eq:XZanticommute9}
	\norm{
		X^{\prime s} Z^{\prime t} \ket{\psi^{\prime}}
		-
		(-1)^{s \cdot t}
		Z^{\prime t} X^{\prime s} \ket{\psi^{\prime}}
	}
	\leq  
    \epsilon_{ac}(s,t)
\end{equation}
and
\begin{equation}
\label{eq:XZswapgeneral9}
	\norm{
	X^{\prime s} \ket{\psi^{\prime}}
	-
	(-1)^{P(s)}
	Z^{\prime \mathbf{A} s} \ket{\psi^{\prime}}
	}
	\leq 
    \epsilon_{xz}(s).
\end{equation}
Then there exists an isometry $\Phi$ and a state $\ket{junk}$ such that for any $p,q \in (0,1)^{n}$
\begin{eqnarray}
\nonumber
    \norm{
        \Phi(
            X^{\prime q}
            Z^{\prime p}
            \ket{\psi^{\prime}}
        )
        -
        \ket{junk}
        X^q
        Z^p
        \ket{\psi}    
    }
    \leq 
\\
\nonumber
\qquad
\qquad
\qquad
\qquad
    \sqrt{
        \frac{1}{2^{2n - 1}}
        \sum_{s,t}
        \epsilon_{ac}(s,p) + 
        \epsilon_{ac}(s,p \oplus t) 
    } 
    +
\\
\qquad
\qquad
\qquad
\qquad
\qquad
\qquad
\qquad
\qquad
    \sqrt{
        \frac{1}{2^{2n - 1}}
        \sum_{t,u}
        \epsilon_{ac}(t,u) + 
        \epsilon_{xz}(u)
    }
    .
\end{eqnarray}
\end{lemma}

\begin{proof}
We first specify the isometry $\Phi$ via a sequence of actions:
\begin{enumerate}
\item Attach $2n$ qubit ancillas with qubit $k$ and $k + n$ in the state $\frac{1}{\sqrt{2}}(\ket{00} + \ket{11})$ for each $k = 1 \dots n$.
\item For $k =  n  \dots 1$ apply a controlled $X^{\prime}_k$ to $\ket{\psi^{\prime}}$, controlled on ancilla qubit $k + n$.
\item Apply Hadamard gates to the last $n$ ancilla qubits
\item For $k =  n  \dots 1$ apply a controlled $Z^{\prime}_k$ to $\ket{\psi^{\prime}}$, controlled on ancilla qubit $k + n$.
\item Apply Hadamard gates to the last $n$ ancilla qubits
\item For $k =  n  \dots 1$ apply a controlled $X^{\prime}_k$ to $\ket{\psi^{\prime}}$, controlled on ancilla qubit $k+ n$.
\end{enumerate}

The state after applying the isometry is straightforwardly found to be
\begin{equation}
\ket{\psi_1} = 
        \Phi(
            X^{\prime q}
            Z^{\prime p}
            \ket{\psi^{\prime}}
        )
        =
    \frac{1}{\sqrt{2^{3n}}}
    \sum_{s,t,u}
    (-1)^{
        t \cdot (s \oplus u)
    }
    X^{\prime u}
    Z^{\prime t}
    X^{\prime s \oplus q}
    Z^{\prime p}
    \ket{\psi^{\prime}}
    \ket{su}.
\end{equation}
We will compare it to the following state:
\begin{equation}
\ket{\psi_3} = 
    \frac{1}{\sqrt{2^{3n}}}
    \sum_{s,t,u}
    (-1)^{
        t \cdot s + 
        p \cdot (u \oplus s \oplus q) +
        P(u \oplus s \oplus q)
    }
    Z^{\prime t \oplus p \oplus \mathbf{A}(u \oplus s \oplus q)}
    \ket{\psi^{\prime}}
    \ket{su}.
\end{equation}

First, let us show that $\ket{\psi_3}$ is normalized.
\begin{eqnarray}
    \braket{\psi_3}{\psi_3} & =  &  
    \frac{1}{2^{3n}}
    \sum_{s, t, t^{\prime}, u}
        (-1)^{
            (t \oplus t^{\prime}) \cdot s
        }
        \bra{\psi^{\prime}}
        Z^{\prime t \oplus t^{\prime}}
        \ket{\psi^{\prime}} \\
    & = & 
    \frac{1}{2^{2n}}
    \sum_{t, t^{\prime}}
        \left(
            \sum_s
                (-1)^{
                    (t \oplus t^{\prime}) \cdot s
                }
        \right)
        \bra{\psi^{\prime}}
        Z^{\prime t \oplus t^{\prime}}
        \ket{\psi^{\prime}} \\        
    & = & 1
\end{eqnarray}
In the first line we omit $s$ and $u$ cross-terms which are zero from $\braket{s^\prime u^\prime}{su}$ factors, and cancel common factors of $Z^{\prime}$.  There is also considerable cancellation in the phases.  To obtain the second line, the summation over $u$ becomes a factor of $2^n$ and we factor out the summation over $s$.  By Lemma~\ref{lemma:stringsums} this factor is $2^n$ when  $t = t^{\prime}$ and $0$ otherwise, so all summands are in fact equal to 1.

With some work we can factor $\ket{\psi_3}$.  First we make two changes of variable, $t \mapsto t \oplus p \oplus \mathbf{A}(s \oplus u)$, followed by $u \mapsto u \oplus q$.  Then
\begin{equation}
\ket{\psi_3} = 
    \frac{1}{\sqrt{2^{3n}}}
    \sum_{s,t,u}
    (-1)^{
        s \cdot (t \oplus \mathbf{A}(s \oplus u)) +
        p \cdot u) +
        P(u \oplus s)
    }
    Z^{\prime t}
    \ket{\psi^{\prime}}
    \ket{s} \ket{u \oplus q}.
\end{equation}
Applying \eref{eq:PR9}, this becomes
\begin{equation}
\ket{\psi_3} = 
    \frac{1}{\sqrt{2^{3n}}}
    \sum_{s,t,u}
    (-1)^{
        s \cdot t +
        p \cdot u +
        P(u) +
        P(s) 
    }
    Z^{\prime t}
    \ket{\psi^{\prime}}
    \ket{s} \ket{u \oplus q}.
\end{equation}
Now we can factor out the summation over $u$.  
\begin{eqnarray}
\nonumber
    \ket{\psi_3} 
& = & 
    \left(
        \frac{1}{\sqrt{2^{2n}}}
        \sum_{s,t}
        (-1)^{
            s \cdot t +
            P(s)
        }
        Z^{\prime t}
        \ket{\psi^{\prime}}
        \ket{s} 
    \right)
    \otimes
    \left(
        \sum_u
        (-1)^{
            p \cdot u +
            P(u)
        }
        \ket{u \oplus q}
    \right)
\\
\\
& = &
    \left(
        \frac{1}{\sqrt{2^{2n}}}
        \sum_{s,t}
        (-1)^{
            s \cdot t +
            P(s)
        }
        Z^{\prime t}
        \ket{\psi^{\prime}}
        \ket{s} 
    \right)
    \otimes
    X^{q} Z^{p}
    \ket{\psi}
\end{eqnarray}
Defining 
$
    \ket{junk}
    =
    \left(
        \frac{1}{\sqrt{2^{2n}}}
        \sum_{s,t}
        (-1)^{
            s \cdot t +
            P(s)
        }
        Z^{\prime t}
        \ket{\psi^{\prime}}
        \ket{s} 
    \right)
$ we find that 
\begin{equation}
    \ket{\psi_3} = \ket{junk} X^q Z^p \ket{\psi}
    .
\end{equation}

Now we wish to estimate the distance between $\ket{\psi_1}$ and $\ket{\psi_3}$.  To do this we will use an intermediate step in the form of the state
\begin{equation}
\ket{\psi_2} = 
\frac{1}{\sqrt{2^{3n}}}
\sum_{s,t,u}
(-1)^{
    t \cdot (u \oplus q)
}
X^{\prime u \oplus s \oplus q}
Z^{\prime t \oplus p}
\ket{\psi^{\prime}}
\ket{su}.
\end{equation}
First we estimate the distance between $\ket{\psi_1}$ and $\ket{\psi_2}$, which we will do by estimating the inner product:
\begin{eqnarray}
\nonumber
\braket{\psi_2}{\psi_1} =
\frac{1}{2^{3n}}
\sum_{s,t,t^{\prime}, u}
    (-1)^{
    t^\prime \cdot (u \oplus q)+
    t \cdot (s \oplus u)
    }
    \\
    \qquad
    \qquad
    \qquad
    \qquad
    \qquad
    \qquad
    \bra{\psi^{\prime}}
        Z^{\prime t^{\prime} \oplus p}
        X^{\prime u \oplus s \oplus q}
        X^{\prime u}
        Z^{\prime t}    
        X^{\prime s \oplus q}
        Z^{\prime p}
    \ket{\psi^{\prime}}.
\end{eqnarray}
Here we have omitted many zero cross terms for $u$ and $s$ resulting from the factor $\braket{su}{s^{\prime}u^{\prime}}$.  Now we can do some cleaning up by the change of variable $s \mapsto s \oplus q$ and cancelling an $X^{\prime u}$ factor.  This further allows us to factor out the sum over $u$.  We then do a further change of variable $u \mapsto u \oplus q$, giving
\begin{equation}
\braket{\psi_2}{\psi_1} =
\frac{1}{2^{3n}}
\sum_{s,t,t^{\prime}}
    \left(
        \sum_u
        (-1)^{
            u \cdot (t \oplus t^{\prime})
        }
    \right)
    (-1)^{
        s \cdot t
    }
    \bra{\psi^{\prime}}
        Z^{\prime t^{\prime} \oplus p}
        X^{\prime s}
        Z^{\prime t}    
        X^{\prime s}
        Z^{\prime p}
    \ket{\psi^{\prime}}.
\end{equation} 
According to Lemma~\ref{lemma:stringsums} we can set $t = t^{\prime}$ since all other terms will be zero.  Let us further make the substitution 
$
    \bra{\psi^{\prime}}
        Z^{\prime t^{\prime} \oplus p}
        X^{\prime s}
        Z^{\prime t}    
        X^{\prime s}
        Z^{\prime p}
    \ket{\psi^{\prime}} 
    =
    (-1)^{s \cdot t}(
    1
    -e(s,t,p)
    )
$, giving
\begin{equation}
\braket{\psi_2}{\psi_1} =
1 - 
\frac{1}{2^{2n}}
\sum_{s,t}
    e(s,t,p)
\end{equation} 
Two applications of \eref{eq:XZanticommute9} give us
\begin{equation}
\norm{
    Z^{\prime t}    
    X^{\prime s}
    Z^{\prime p}
    \ket{\psi^{\prime}} 
    -
    (-1)^{
    s \cdot t
    }
    X^{\prime s}
    Z^{\prime t \oplus p}
    \ket{\psi^{\prime}}     
} \leq
\epsilon_{ac}(s,p) + 
\epsilon_{ac}(s,p \oplus t).
\end{equation}
Multiplying on the left by the norm-1 operator 
$
    \bra{\psi^{\prime}}
    Z^{\prime t \oplus p}
    X^{\prime s}
$ allows us to bound 
$
    |e(s,t,p)| \leq 
    \epsilon_{ac}(s,p) + 
    \epsilon_{ac}(s,p \oplus t) 
$ 
and hence by Lemma~\ref{lemma:innerproducttonorm}
\begin{equation}
\norm{
    \ket{\psi_1}
    -
    \ket{\psi_2}
}
\leq
\sqrt{
    \frac{1}{2^{2n - 1}}
    \sum_{s,t}
    \epsilon_{ac}(s,p) + 
    \epsilon_{ac}(s,p \oplus t) 
}
\end{equation}

Next we estimate the distance between $\ket{\psi_2}$ and $\ket{\psi_3}$.  After dropping zero cross-terms for $s$ and $u$ we obtain
\begin{eqnarray}
\nonumber
    \braket{\psi_2}{\psi_3} = 
    \frac{1}{2^{3n}}
    \sum_{s, t, t^{\prime}, u}
        (-1)^{
            t^{\prime} \cdot (u \oplus q) + 
            t \cdot s +
            p \cdot (s \oplus u \oplus q) +
            P(u \oplus s \oplus q)
        }
        \\
        \qquad
        \qquad
        \qquad
        \qquad
        \qquad
        \qquad
        \qquad
        \bra{\psi^{\prime}}
            Z^{\prime t^{\prime} \oplus p}        
            X^{\prime u \oplus s \oplus q}
            Z^{\prime t \oplus p \oplus \mathbf{A}(u \oplus s \oplus q)} 
        \ket{\psi^{\prime}}
\end{eqnarray}
We make changes of variable $u \mapsto s \oplus u \oplus q$ to clean things up, after which we can factor out the sum over $s$, giving
\begin{eqnarray}
\nonumber
    \braket{\psi_2}{\psi_3} = 
    \frac{1}{2^{3n}}
    \sum_{t, t^{\prime},u}
        \left(
        \sum_s
            (-1)^{
            s\cdot (t^{\prime} \oplus t)
            }
        \right)
        (-1)^{
            u \cdot (t^{\prime} \oplus p) + 
            P(u)
        }
\\
\qquad
\qquad
\qquad
\qquad
\qquad
\qquad
\qquad
\qquad
        \bra{\psi^{\prime}}
            Z^{\prime t^{\prime} \oplus p}        
            X^{\prime u}
            Z^{\prime t \oplus p \oplus \mathbf{A}u} 
        \ket{\psi^{\prime}}
\end{eqnarray}
Applying Lemma~\ref{lemma:stringsums}, we can drop all terms except where $t = t^{\prime}$.  We further make the change of variable $t \mapsto t \oplus p$ and the substitution $
        \bra{\psi^{\prime}}
            Z^{\prime t}        
            X^{\prime u}
            Z^{\prime t \oplus \mathbf{A}u} 
        \ket{\psi^{\prime}}
        = (-1)^{
            u \cdot t +
            P(u)
        }
        (1
        - 
        f(t,u)
        )
$ to obtain
\begin{equation}
    \braket{\psi_2}{\psi_3} = 
    1 - 
    \frac{1}{2^{2n}}
    \sum_{t,u}
        f(t,u).
\end{equation}
We can estimate $f(t,u)$ by first using \eref{eq:XZanticommute9}and then \eref{eq:XZswapgeneral9} to find
\begin{equation}
\norm{
        X^{\prime u}
        Z^{\prime t}        
        \ket{\psi^{\prime}}
        -
        (-1)^{
            u \cdot t +
            P(u)
        }
        Z^{\prime t \oplus \mathbf{A}u} 
        \ket{\psi^{\prime}}
}
\leq
\epsilon_{ac}(t,u) +
\epsilon_{xz}(u).
\end{equation}
Then multiplying on the left by the norm-1 operator 
$
    \bra{\psi^{\prime}}
        Z^{\prime t}        
        X^{\prime u}
$
gives us 
$
f(t,u)
\leq
\epsilon_{ac}(t,u) +
\epsilon_{xz}(u)
$
and then by Lemma~\ref{lemma:innerproducttonorm}
\begin{equation}
\norm{
\psi_2 - \psi_3}
\leq
\sqrt{
    \frac{1}{2^{2n - 1}}
    \sum_{t,u}
    \epsilon_{ac}(t,u) + 
    \epsilon_{xz}(u)
}
\end{equation}

The triangle inequality allows us to estimate $\norm{\ket{\psi_1} - \ket{\psi_3}}$ and gives us our desired result.
\end{proof}

The isometry and proof are very similar to that in \cite{McKague:2013:Interactive-proofs-for-BQP-via-self-tested-graph-states}, but there are some important differences.  Here, we do not have any assumptions about whether the operators commute.  This means that in general $\Phi$ cannot be factored into a tensor product, or even a product of commuting isometries.  However, if some commutation relations are known between operators, it may be possible to factor $\Phi$.  The proof here is slightly rearranged, requiring only one intermediate step instead of three, and of course phrased for maximum generality.

Here we have in mind that the matrix $\mathbf{A}$ is the adjacency matrix of some graph and $\ket{\psi}$ is the corresponding graph state.  In this case, setting $P(s) = \left(\frac{1}{2} s \cdot \mathbf{A} s \right) \pmod{2}$ has the required property, as shown in~\cite{McKague:2013:Interactive-proofs-for-BQP-via-self-tested-graph-states} Lemma~2.

\subsection{Sufficient conditions for self-testing many maximally entangled pairs of qubits}
The goal of the following lemma is to reduce the conditions of Lemma~\ref{lemma:graphstateselftestconditions} to something easier to deal with.  In Lemma~\ref{lemma:graphstateselftestconditions} we needed to consider products of up to $n$ operators, but any one measurement will provide direct information about products of one operator on Alice's side and one on Bob's side.  Thus we need to use the direct information from measurements to learn about the products required for Lemma~\ref{lemma:graphstateselftestconditions}.

Note that $R$ is the adjacency matrix for $\frac{n}{2}$ isolated edges.  The corresponding graph state is then $\frac{n}{2}$ maximally entangled pairs of qubits.

\begin{lemma}
\label{lemma:acandxz}
Suppose that 
\begin{enumerate}
\item $\ket{\psi^{\prime}} \in \mathcal{H}_A \otimes \mathcal{H}_B$ is a state 
\item $\{X^{\prime}_k\}_{k=1}^{\frac{n}{2}}$  are unitary, Hermitian, pair-wise commuting operators on $\mathcal{H}_A$ 
\item $\{Z^{\prime}_k\}_{x=1}^{\frac{n}{2}}$  are unitary, Hermitian, pair-wise commuting operators on $\mathcal{H}_A$ 
\item $\{X^{\prime}_k\}_{k=\frac{n}{2} + 1}^{n}$  are unitary, Hermitian, pair-wise commuting operators on $\mathcal{H}_B$ 
\item $\{Z^{\prime}_k\}_{k=\frac{n}{2} + 1}^{n}$  are unitary, Hermitian, pair-wise commuting operators on $\mathcal{H}_B$ 
\end{enumerate}
such that for all $k \neq \ell$
\begin{eqnarray}
\label{eq:XZapproxcommute2}
	\norm{
		X^{\prime}_k Z^{\prime}_\ell \ket{\psi^{\prime}}
		-
		Z^{\prime}_\ell X^{\prime}_k \ket{\psi^{\prime}}
	}
	& \leq & \epsilon_1
\\
\label{eq:XnearZ2}
	\norm{
		X^{\prime}_k \ket{\psi^{\prime}}
		-
		Z^{\prime}_{k + \frac{n}{2}} \ket{\psi^{\prime}}
	}
	& \leq & \epsilon_2  
\\
\label{eq:XZanticommute2}
	\norm{
		Z^{\prime}_k X^{\prime}_k \ket{\psi^{\prime}}
		+ 
		X^{\prime}_k Z^{\prime}_k \ket{\psi^{\prime}}
	}
	& \leq & \epsilon_3
\end{eqnarray}
where $k+\frac{n}{2}$ is taken modulo $n$. Then
\begin{enumerate}
\item 
for any $s,t \in (0,1)^{2n}$
\begin{eqnarray}
\nonumber
	\norm{
		X^{\prime s} Z^{\prime t} \ket{\psi^{\prime}}
		-
		(-1)^{s \cdot t}
		Z^{\prime t} X^{\prime s} \ket{\psi^{\prime}}
	}
	\leq  
\\
\nonumber
    \qquad\qquad\qquad
    (|s_a||t_a| + |s_b||t_b|) (\epsilon_1 + 2 \epsilon_2) +    
\\
\qquad\qquad\qquad\qquad
\label{eq:XZanticommutegeneral}    
     (t \cdot s) (\epsilon_3 - \epsilon_1) + 2 \epsilon_2 \min \{|s|, |t| \} 
    =:
    \epsilon_{ac}(s,t)
\end{eqnarray}

\item
for any $s \in (0,1)^{2n}$
\begin{eqnarray}
\nonumber
	\norm{
	X^{\prime s} \ket{\psi^{\prime}}
	-
	(-1)^{Rs_a \cdot s_b}
	Z^{\prime Rs} \ket{\psi^{\prime}}
	}
	\leq 
\\ 
\qquad
\qquad
\qquad
\qquad
	|s|
	\epsilon_2
	+
	\epsilon_{ac}(s_a, Rs_b) +
    \epsilon_{ac}(s_b, Rs_a) =: \epsilon_{xz}(s).
\end{eqnarray}
\end{enumerate}

\end{lemma}

For (i) we need to use the anti-commutation and commutation estimate many times in order to exchange the position of the $X^{\prime}$ and $Z^{\prime}$ operators.  For (ii) we additionally need to use the correlation estimates many times to change $X^{\prime}$ operators into $Z^{\prime}$ operators.  However, we can only use the estimates if the operators in question are rightmost in the product of operators.  Hence we must take advantage of the fact that many of the operators exactly commute to move two operators of interest to the right so we can apply an estimate.  The ordering of these moves is quite sensitive.  Before proving the lemma we will first need to prove the following claim which will help us with the reorderings:
\begin{claim}
\label{claim:xzswapstep}
Let $k \in \{1, \dots, n\}$ and $t \in \{0,1\}^n$ such that either $k \leq \frac{n}{2}$ and $t = t_a$, or $k > \frac{n}{2}$ and $t = t_b$.  Then under the conditions of Lemma~\ref{lemma:acandxz}
\begin{equation}
\label{eq:magicstep1}
	\norm{
		Z^{\prime t} X^{\prime}_k \ket{\psi^{\prime}}
		- 
		(-1)^{t_k}
		X^{\prime}_k Z^{\prime t} \ket{\psi^{\prime}}
	}
	\leq 
	|t| (\epsilon_1 + 2\epsilon_2) + t_k (\epsilon_3 - \epsilon_1)
\end{equation}
and
\begin{equation}
\label{eq:magicstep2}
	\norm{
		X^{\prime t} Z^{\prime}_k \ket{\psi^{\prime}}
		- 
		(-1)^{t_k}
		Z^{\prime}_k X^{\prime t} \ket{\psi^{\prime}}
	}
	\leq 
	|t| (\epsilon_1 + 2\epsilon_2) + t_k (\epsilon_3 - \epsilon_1)
\end{equation}
\end{claim}

\begin{proof}
We will consider the case where $k \leq \frac{n}{2}$ and $t = t_a$.  The other case follows analogously.  The general idea will be to take advantage of the fact that all operators are on Alice's side.  We can move a single operator to Bob's side using an estimation, after which it will freely commute past all other operators.

Begin by defining $i_m$ to be the $m$th index $\ell$ such that $t_\ell = 1$.  Thus
\begin{equation}
	Z^{\prime t} = Z^{\prime}_{i_1} \dots Z^{\prime}_{i_{|t|}}
\end{equation}
Note that $i_m \leq \frac{n}{2}$ for all $m$.  For now, let us suppose that $t_k = 0$ so that $i_m$ is never equal to $k$.  We move $X^{\prime}_{k}$ to the left one position using \eref{eq:XZapproxcommute2}:
\begin{equation}
	\norm{
		Z^{\prime}_{i_1} \dots Z^{\prime}_{i_{|t|}} X^{\prime}_k \ket{\psi^{\prime}}
		- 
		Z^{\prime}_{i_1} \dots Z^{\prime}_{i_{|t| - 1}} X^{\prime}_{k} Z^{\prime}_{i_{|t|}}\ket{\psi^{\prime}}
	}
	\leq
	\epsilon_1
\end{equation}
Now $Z^{\prime}_{i_{|t|}}$ can then be moved over to Bob's side using \eref{eq:XnearZ2}, giving
\begin{equation}
	\norm{
		Z^{\prime}_{i_1} \dots Z^{\prime}_{i_{|t|}} X^{\prime}_k \ket{\psi^{\prime}}
		- 
		Z^{\prime}_{i_1} \dots Z^{\prime}_{i_{|t| - 1}} X^{\prime}_{k} X^{\prime}_{i_{|t|} + \frac{n}{2}}\ket{\psi^{\prime}}
	}
	\leq
	\epsilon_1 + \epsilon_2.
\end{equation}
We can now move $X^{\prime}_{i_{|t|} + \frac{n}{2}}$ all the way to the left since all other operators are on Alice's subsystem.  We repeat the above sequence $|t|$ times for each $Z^{\prime}_{i_{j}}$, resulting in
\begin{equation}
	\norm{
		Z^{\prime}_{i_1} \dots Z^{\prime}_{i_{|t|}} X^{\prime}_k \ket{\psi^{\prime}}
		- 
		X^{\prime}_{i_{|t|} + \frac{n}{2}} \dots X^{\prime}_{i_{1} + \frac{n}{2}} X^{\prime}_{k}\ket{\psi^{\prime}}
	}
	\leq
	|t|(\epsilon_1 + \epsilon_2).
\end{equation}
$X^{\prime}_k$ moves to the left past all other operators, since it is the only operator left on Alice's system.  
\begin{equation}
	\norm{
		Z^{\prime}_{i_1} \dots Z^{\prime}_{i_{|t|}} X^{\prime}_k \ket{\psi^{\prime}}
		- 
		 X^{\prime}_{k} X^{\prime}_{i_{|t|} + \frac{n}{2}} \dots X^{\prime}_{i_{1} + \frac{n}{2}}\ket{\psi^{\prime}}
	}
	\leq
	|t|(\epsilon_1 + \epsilon_2).
\end{equation}
One by one, we transfer each $X^{\prime}$ on Bob's system operator back to a $Z^{\prime}$ on Alice's system and move it to the left.  This gives
\begin{equation}
	\norm{
		Z^{\prime}_{i_1} \dots Z^{\prime}_{i_{|t|}} X^{\prime}_k \ket{\psi^{\prime}}
		- 
		 X^{\prime}_{k} Z^{\prime}_{i_1} \dots Z^{\prime}_{i_{|t|}}\ket{\psi^{\prime}}
	}
	\leq
	|t|(\epsilon_1 + 2\epsilon_2)
\end{equation}
If $t_k = 1$, then we must make an adjustment: when $i_m = k$ we must swap positions of $X^{\prime}_k$ and $Z^{\prime}_k$ using \eref{eq:XZanticommute2} instead of \eref{eq:XZapproxcommute2}.  The result is a phase of $-1$, and an error of $\epsilon_3$ instead of $\epsilon_1$, giving
\begin{equation}
	\norm{
		Z^{\prime}_{i_1} \dots Z^{\prime}_{i_{|t|}} X^{\prime}_k \ket{\psi^{\prime}}
		+
		X^{\prime}_{k} Z^{\prime}_{i_1} \dots Z^{\prime}_{i_{|t|}}\ket{\psi^{\prime}}
	}
	\leq
	|t|(\epsilon_1 + 2\epsilon_2) + \epsilon_3 - \epsilon_1.
\end{equation}

Combining the two cases above gives~\eref{eq:magicstep1}.  Noting that the conditions of the lemma are symmetric in the roles of $X$ and $Z$, we can run the entire argument again with $X$ and $Z$ swapped to obtain~\eref{eq:magicstep2}.

\end{proof}

\begin{proof}[Proof of Lemma~\ref{lemma:acandxz}]
Noting that $X^{\prime s_a}$ and $Z^{\prime t_a}$ are operators on Alice's system while $X^{\prime s_b}$ and $Z^{\prime t_b}$ are on Bob's system, 
\begin{equation}
X^{\prime s} Z^{\prime_t} = 
X^{\prime {s_b}} Z^{\prime {t_b}}
X^{\prime {s_a}} Z^{\prime {t_a}}.
\end{equation}
Let $k \in \{1 \dots n\}$ be the smallest such that $(s_a)_k = 1$.  We apply Claim~\ref{claim:xzswapstep} to obtain
\begin{eqnarray}
\nonumber
	\norm{
        X^{\prime s} Z^{\prime_t}
        \ket{\psi^{\prime}}
		- 
		(-1)^{t_k}
		X^{\prime {s_b}} Z^{\prime {t_b}}
		X^{\prime {s_a \oplus 1_k}}		
		Z^{\prime t_a} 
		X^{\prime}_k \ket{\psi^{\prime}}
	}
    \leq
\\
\qquad\qquad\qquad\qquad\qquad\qquad\qquad    
	|t_a| (\epsilon_1 + 2 \epsilon_2) + t_k (\epsilon_3 - \epsilon_1).
\end{eqnarray}
Next we apply \eref{eq:XZapproxcommute2} to move $X^{\prime}_k$ to $Z^{\prime}_{k + \frac{n}{2}}$ on Bob's side, which is then commuted to the left, giving
\begin{eqnarray}
\nonumber
	\norm{
        X^{\prime s} Z^{\prime_t}
        \ket{\psi^{\prime}}
		- 
		(-1)^{t_k}
		X^{\prime {s_b}} Z^{\prime {t_b}}
		Z^{\prime}_{k + \frac{n}{2}}	
		X^{\prime {s_a \oplus 1_k}}		
		Z^{\prime t_a} 
		\ket{\psi^{\prime}}	
	}
	\leq 
\\
\qquad
\qquad\qquad\qquad\qquad\qquad\qquad
    |t_a| (\epsilon_1 + 2 \epsilon_2) + t_k (\epsilon_3 - \epsilon_1) + \epsilon_2.
\end{eqnarray}
Repeating this for each position where $s_a = 1$ we get
\begin{eqnarray}
\nonumber
	\norm{
        X^{\prime s} Z^{\prime_t}
        \ket{\psi^{\prime}}
		- 
		(-1)^{t_a \cdot s_a}
		X^{\prime {s_b}} Z^{\prime {t_b}}
		Z^{\prime Rs_a}		
		Z^{\prime t_a} 
		\ket{\psi^{\prime}}	
	}
	\leq 
\\
\qquad \qquad \qquad\qquad\qquad
    |s_a||t_a| (\epsilon_1 + 2 \epsilon_2) + (t_a \cdot s_a) (\epsilon_3 - \epsilon_1) + |s_a|\epsilon_2
\end{eqnarray}
We can commute $Z^{\prime Rs_a}$ past $Z^{\prime t_a}$ since they are on different systems, and then apply~\eref{eq:XZapproxcommute2} $|s_a|$ more times.  Finally, we commute all the operators on Alice's side to the left past Bob's operators to obtain
\begin{eqnarray}
\nonumber
	\norm{
        X^{\prime s} Z^{\prime_t}
        \ket{\psi^{\prime}}
		- 
		(-1)^{t_a \cdot s_a}
		Z^{\prime t_a} 
		X^{\prime s_a}				
		X^{\prime {s_b}} Z^{\prime {t_b}}
		\ket{\psi^{\prime}}	
	}
	\leq 
\\
    \qquad
    \qquad
    \qquad
    \qquad
    \qquad
    |s_a||t_a| (\epsilon_1 + 2 \epsilon_2) + (t_a \cdot s_a) (\epsilon_3 - \epsilon_1) + 2|s_a|\epsilon_2.
\end{eqnarray}
Applying the whole procedure again, swapping the roles of Alice and Bob, we find
\begin{eqnarray}
\nonumber
	\norm{
		X^{\prime s} Z^{\prime t} \ket{\psi^{\prime}}
		-
		(-1)^{s \cdot t}
		Z^{\prime t} X^{\prime s} \ket{\psi^{\prime}}
	}
	\leq
\\  
\qquad\qquad\qquad
     (|s_a||t_a| + |s_b||t_b|) (\epsilon_1 + 2 \epsilon_2) + (t \cdot s) (\epsilon_3 - \epsilon_1) + 2|s|\epsilon_2.
\end{eqnarray}

Since the conditions of the lemma are symmetric in the roles of $X$ and $Z$, we can obtain a similar result with error bounded by $(|s_a||t_a| + |s_b||t_b|) (\epsilon_1 + 2 \epsilon_2) + (t \cdot s) (\epsilon_3 - \epsilon_1) + 2|t|\epsilon_2$ instead.  Taking the minimum we obtain~\eref{eq:XZanticommutegeneral}.



Now we turn our attention to part ii.  Let $k \in \{1 \dots \frac{n}{2}\}$ be the first index such that $s_k = 1$.  Initially, $X^{\prime}_k$ commutes past everything to the right.  Then we can apply \eref{eq:XnearZ2} to obtain
\begin{equation}
	\norm{
	X^{\prime s} \ket{\psi^{\prime}}
	-
	X^{\prime s \oplus 1_k}
	Z^{\prime}_{k+\frac{n}{2}} \ket{\psi^{\prime}}
	}
	\leq
	\epsilon_2
	.
\end{equation}
Applying this in turn for each $k$ such that $(s_a)_k = 1$, we find
\begin{equation}
	\norm{
	X^{\prime s} \ket{\psi^{\prime}}
	-
	X^{\prime s_b}
	Z^{\prime R s_a} \ket{\psi^{\prime}}
	}
	\leq
	|s_a|
	\epsilon_2
	.
\end{equation}
At this point we have a problem, since the remaining $X^{\prime}_k$ are on Bob's side, along with all of the $Z^{\prime}$.  Applying \eref{eq:XZanticommutegeneral}, however, we find
\begin{equation}
	\norm{
	X^{\prime s} \ket{\psi^{\prime}}
	-
	(-1)^{Rs_a \cdot s_b}
	Z^{\prime R s_a} 
	X^{\prime s_b}
	\ket{\psi^{\prime}}
	}
	\leq
	|s_a|
	\epsilon_2
	+
    \epsilon_{ac}(s_a, Rs_b).
\end{equation}
Now we can continue to turn $X^{\prime}_k$ into $Z^{\prime}_{k+\frac{n}{2}}$ for the remaining $k$ such that $(s_b)_k = 1$.  The final result is
\begin{equation}
	\norm{
	X^{\prime s} \ket{\psi^{\prime}}
	-
	(-1)^{Rs_a \cdot s_b}
	Z^{\prime R s} 
	\ket{\psi^{\prime}}
	}
	\leq
	|s|
	\epsilon_2
	+
	\epsilon_{ac}(s_a, Rs_b) +
    \epsilon_{ac}(s_b, Rs_a)
\end{equation}

\end{proof}

Now we are ready to combine Lemmas~\ref{lemma:graphstateselftestconditions} and \ref{lemma:acandxz} to prove sufficient conditions for self-testing many maximally entangled pairs of qubits.

\begin{lemma}
\label{lemma:paralleltestsufficientconditions}
Under the conditions of Lemma~\ref{lemma:acandxz} there exists an isometry $\Phi$ and a state $\ket{junk}$ such that for any $p,q \in (0,1)^{n}$
\begin{eqnarray}
\nonumber
    \norm{
        \Phi(
            X^{\prime q}
            Z^{\prime p}
            \ket{\psi^{\prime}}
        )
        -
        \ket{junk}
        X^q
        Z^p
        \ket{\psi}    
    }
    \leq 
\\
\nonumber
\qquad
\qquad
    \sqrt{
        \frac{|p|}{2}
        \left(
            (n-1) \epsilon_1 +
            2n \epsilon_2 +
            \epsilon_3
        \right)
        + 
        \frac{n}{4}
        \left(
            \epsilon_3 -
            \epsilon_1
        \right)
        +
        \frac{n^2}{8}
        \left(
            \epsilon_1 +
            2\epsilon_2
        \right)
    } 
    +
\\
\qquad
\qquad
\qquad
\qquad
\qquad
    \sqrt{
        \frac{n^2}{4}
        \left(
            \epsilon_1 + 
            2 \epsilon_2
        \right)
        +
        \frac{n}{2}
        \left(
            \epsilon_2 +
            \epsilon_3 - 
            \epsilon_1
        \right)
    }
    .
\end{eqnarray}

\end{lemma}

\begin{proof}
The conditions of Lemma~\ref{lemma:acandxz}, together with Lemma~\ref{lemma:Rdecompose} straightforwardly imply the conditions of Lemma~\ref{lemma:graphstateselftestconditions}.  All that remains is to substitute in the error estimates to calculate the bound.  This is done using Lemma~\ref{lemma:stringsums} repeatedly.
\end{proof}

\section{Parallelizing the Mayers-Yao test}
\label{sec:ParallelizingtheMayersYaotest}
We want to test $\frac{n}{2}$ e-bits shared between Alice and Bob, so we will devise some test that allows us, for honest players, to fullfil the conditions of Lemma~\ref{lemma:paralleltestsufficientconditions}.  This test will be based on the Mayers-Yao test for a single e-bit.

Let us first summarize the conditions of Lemma~\ref{lemma:paralleltestsufficientconditions}.  We need some $X^{\prime}$ operators on Alice's system that commute with each other, and similar for Bob.  We also need some $Z^{\prime}$ operators that commute with each other for Alice and some more for Bob.  These are easy enough to construct, and the commutation properties will follow directly from the construction.  Next we need Alice's $X^{\prime}_j$ measurements to correlate with Bob's $Z^{\prime}_j$ measurements, which we can test by asking Alice and Bob to perform those measurements and checking their answers.  We also need Alice's $X^{\prime}_j$ measurement to anti-commute with her $Z^{\prime}_j$ measurment.  This is established using measurements from the Mayers-Yao test.

So far, everything is very similar to other self-tests (such as \cite{McKague:2012:Robust-self-tes} and \cite{McKague:2010:Self-testing-gr}), in which the most difficult part is to show that we have anti-commuting operators.  Indeed, the conditions outlined so far are sufficient to show that there is at least \emph{one} maximally entangled pair of qubits.  However, we do not yet have any way of guaranteeing that Alice's system $j$ is independent of her system $k$ (for $j \neq k$), so there might be \emph{only} one pair.  Here we will guarantee independence by the fact that Alice's $X^{\prime}_j$ \emph{commutes} with $Z^{\prime}_k$ for $j \neq k$.  Interestingly, this will be in some sense the most difficult part, with all but a constant number of measurement settings being dedicated to establishing this property.  While in \cite{McKague:2012:Robust-self-tes} and \cite{McKague:2010:Self-testing-gr} we need to add measurements to show that some things anti-commute, here we add measurements to show that some things commute.

\subsection{Structure of test and honest behaviour}
It will be convenient to specify both the structure of the test and the honest behaviour together.  Alice and Bob will start off sharing the following state, which is the graph state of $\frac{n}{2}$ isolated edges:
\begin{equation}
	\ket{\psi} = 
    \frac{1}{2}
    \left(
        \ket{00} + \ket{01} + \ket{10} - \ket{11}
    \right)^{\otimes \frac{n}{2}}.
\end{equation}
Here Alice holds the first qubit of each pair, and Bob holds the second, so that Alice's qubit $k$ is entangled with Bob's qubit $k$.

The possible questions (measurement settings) for Alice are:
\begin{enumerate}
\item[$A_X$] Alice measures each of her $\frac{n}{2}$ qubits in the eigenbasis of $X$ and returns the results.
\item[$A_Z$] Alice measures each of her $\frac{n}{2}$ qubits in the eigenbasis of $Z$ and returns the results.
\item[$A_D$] Alice measures each of her $\frac{n}{2}$ qubits in the eigenbasis of $D = \frac{X + Z}{\sqrt{2}}$ and returns the results.
\item[$A_{Xj}$] For each $k = 1 \dots \frac{n}{2}$, Alice measures her $k$th qubit in basis $X$ if the $j$th bit of $k$ (as a binary number) is 1.  Otherwise she measures in basis $Z$.  She returns the results.
\item[$A_{Zj}$] For each $k = 1 \dots \frac{n}{2}$, Alice measures her $k$th qubit in basis $Z$ if the $j$th bit of $k$ is 0.  Otherwise she measures in basis $X$.  She returns the results.
\end{enumerate}

Bob's questions and behaviour are analogous. For our convenience, the players will return their measurement results as the eigenvalue, i.e., $\pm 1$.

The referee will never choose the question $A_D B_D$ or any combination of $A^{X_j}$ or ${A^{Z_j}}$ with $B^{X_\ell}$ or ${B^{Z_\ell}}$ since we will not use any of these measurements.

\subsection{General behaviour}
\label{sec:generalbehaviour}

Each player receives one question and returns $\frac{n}{2}$ answers as a string in $\{-1, 1\}^\frac{n}{2}$.  The most general behaviour for the players to share some joint state $\ket{\psi^{\prime}}$ and each player performs a POVM on their subsystem and returns the result.  Since we are not concerned with the dimension of the players' systems, we can perform a Steinspring dilation and turn the POVM into a projective measurement.  Note that the dilation is an isometry on the player's system, and so it can be absorbed into our definition of $\Phi$ (the composition of two isometries is again an isometry).  Also, if the state is mixed we can include the purification in the state, so it is not a restriction to assume that it is pure.\footnote{
Another way of dealing with mixed states is to introduce a third register to hold the purification.  Then Alice and Bob's measurements will never touch the purification.  From the construction it is obvious that $\Phi$ will never touch this third register, nor will the operators obtained by applying $\Phi$ to Alice and Bob's measurements.  We can then know for sure that the purification register ends up in $\ket{junk}$ and that no part of the self-tested state resides in the purification.
}

Now a player's behaviour can be modelled as a collection of projective measurements:
\begin{equation}
	\mathcal{M}_q = \{ \Pi^q_a\}_a
\end{equation}
where $q$ is the question, $a$ is a string of answers and $\Pi^q_a \Pi^q_b = \delta_{a,b} \Pi^{q}_a$.  We can define projectors for individual symbols in the answer as follows:
\begin{equation}
	\Gamma^q_{k,x} = \sum_{a: a_k = x} \Pi^q_a
\end{equation}
where $k \in \{1 \dots \frac{n}{2}\}$, $a_k$ is the $k$-th symbol of $a$, and $x = \pm 1$.  Note that for all $j,k, x, y$ the operators $\Gamma^q_{j,x}$ and $\Gamma^q_{k,y}$ commute since $\Pi^q_a$ and $\Pi^q_b$ commute for all $a,b$.  We can next define observables for each answer symbol as
\begin{equation}
M^{\prime q}_k = \Gamma^q_{k,1} - \Gamma^q_{k, -1}
.
\end{equation}
Note that $M^{\prime q}_k$ and $M^{\prime r}_\ell$ will commute whenever $q = r$ (by construction), or when $q$ is a question for Alice and $r$ is a question for Bob (since the operators are defined on two different subsystems).  Now measuring $\mathcal{M}_{q}$ is equivalent to measuring $M^{\prime q}_k$ for each $k$ and returning the resulting eigenvalues as a string.  Each $M^{\prime q}_k$ is Hermitian and unitary.

We will give some more convenient names for some of these operators. 
\begin{equation}
    X^{\prime}_k  = M^{\prime X}_k
    \cases{
        M^{\prime A_X}_k & $0 < k \leq \frac{n}{2}$\\
        M^{\prime B_X}_{k - \frac{n}{2}} & $\frac{n}{2} < k \leq n$ \\        
    }
\end{equation}

\begin{equation}
    Z^{\prime}_k  = M^{\prime Z}_k
    \cases{
        M^{\prime A_Z}_k & $0 < k \leq \frac{n}{2}$\\
        M^{\prime B_Z}_{k - \frac{n}{2}} & $\frac{n}{2} < k \leq n$ \\        
    }
\end{equation}

\begin{equation}
    M^{\prime Xj}_k  = 
    \cases{
        M^{\prime A_{Xj}}_k & $0 < k \leq \frac{n}{2}$\\
        M^{\prime B_{Xj}}_{k - \frac{n}{2}} & $\frac{n}{2} < k \leq n$ \\        
    }
\end{equation}

\begin{equation}
    M^{\prime Zj}_k  = 
    \cases{
        M^{\prime A_{Zj}}_k & $0 < k \leq \frac{n}{2} $\\
        M^{\prime B_{Zj}}_{k - \frac{n}{2}} & $\frac{n}{2} < k \leq n$ \\        
    }
\end{equation}

For convenience we will also define operators $M^{Xj}_k$ etc.\ (i.e., without the $\prime$) for the honest behaviour, which will be one of the Pauli operators $X$ or $Z$, or the sum $\frac{X + Z}{\sqrt{2}}$.

\subsection{Proof of self-testing}
\begin{lemma} 
\label{lemma:testproperties}
Given the definitions of operators in section~\ref{sec:generalbehaviour},
\begin{enumerate}
    \item 
        for all $k, \ell$
        \begin{eqnarray}
        \label{eq:Xcommutes}
                X^{\prime}_\ell X^{\prime}_k = X^{\prime}_k X^{\prime}_\ell \\        
        \label{eq:Zcommutes}        
                Z^{\prime}_\ell Z^{\prime}_k = Z^{\prime}_k Z^{\prime}_\ell
        \end{eqnarray}
    \item 
        when $k \leq \frac{n}{2}$, $\ell > \frac{n}{2}$ or  $\ell \leq \frac{n}{2}$, $k > \frac{n}{2}$
        \begin{equation}
            \label{eq:XZcommutes}
                X^{\prime}_\ell Z^{\prime}_k = Z^{\prime}_k X^{\prime}_\ell
        \end{equation}
\end{enumerate}
If additionally for each $q,r \in \{ X, Z\} \cup \{X_j\}_j \cup \{Z_j\}_j$ and $0 < k \leq \frac{n}{2}$ we have
\begin{equation}
\left|
    \bra{\psi^\prime} 
        M^{\prime q}_k
        M^{\prime r}_{k + \frac{n}{2}}
        \ket{\psi^\prime} -
    \bra{\psi} 
        M^{q}_k
        M^{r}_{k + \frac{n}{2}}
    \ket{\psi}
\right|
\leq
\epsilon
\end{equation}
then
\begin{enumerate}
    \item[(iii)]
        for all $k$
        \begin{eqnarray}
        \label{eq:XnearZ}
            \norm{
                X^{\prime}_k \ket{\psi^{\prime}}
                -
                Z^{\prime}_{k + \frac{n}{2}}\ket{\psi^{\prime}}	
            }
            & \leq &
            \sqrt{2 \epsilon}
        \end{eqnarray}
        (where $k + \frac{n}{2}$ is taken modulo $2n$)

    \item[(iv)] 
        for $k \neq l$ and $k, \ell \leq \frac{n}{2}$ or $k, \ell > \frac{n}{2}$
        \begin{equation}
        \label{eq:XZapproxcommute}
            \norm{
                X^{\prime}_\ell Z^{\prime}_k \ket{\psi^{\prime}} -
                Z^{\prime}_k X^{\prime}_\ell \ket{\psi^{\prime}}
            }
            \leq  
            4 \sqrt{2 \epsilon}
        \end{equation}
\end{enumerate}
\end{lemma}

\begin{proof}
Equations \eref{eq:Xcommutes} and \eref{eq:Zcommutes} follow directly from the construction of 
$
    X^\prime_j
$ and
$  
    Z^\prime_j
$.
For \eref{eq:XZcommutes}, if one of $k$ and $\ell$ is less than or equal to $\frac{n}{2}$, and the other is greater than $\frac{n}{2}$, then $X^{\prime}_\ell$ and $Z^{\prime}_k$ are on different subsystems, so they necessarily commute. 

For \eref{eq:XnearZ}, suppose $k \leq \frac{n}{2}$, with the other case following analogously.  The honest behaviour for Alice is to measure her $k$th qubit in the $X$ basis.  Bob's honest behaviour is to measure his $k$th qubit in the $Z$ basis.  Hence
\begin{equation}
\bra{\psi} X_k Z_{k + \frac{n}{2}} \ket{\psi} = 1
\end{equation}
and by assumption we then have
\begin{equation}
	\bra{
		\psi^{\prime}
	} 
		X^{\prime}_k Z^{\prime}_{k + \frac{n}{2}}
	\ket{
		\psi^{\prime}
	} \geq 1 - \epsilon
	.
\end{equation}
This readily translates into \eref{eq:XnearZ} using Lemma~\ref{lemma:innerproducttonorm}. 

We now consider \eref{eq:XZapproxcommute}.  Let us suppose that $k, \ell \leq \frac{n}{2}$.  The other case follows analogously.  We begin by looking at the binary representation of the numbers $k$ and $\ell$.  Since $k \neq \ell$ there is some position $j$ where their binary representations disagree.  Let us suppose that the $j$th bit of $k$ is 0 and the $j$th bit of $\ell$ is 1, with the other case following analogously. Consider the measurement $Z^{\prime}_{k} M^{\prime Xj}_{k+\frac{n}{2}} $.  In the honest case Alice would measure $Z$ on her $k$th qubit, and Bob would measure $X$ on his $k$th qubit, giving $\bra{\psi}Z_{k} M^{Xj}_{k+\frac{n}{2}}\ket{\psi} = 1$.  Following the argument for \eref{eq:XnearZ} we find
\begin{equation}
	\norm{
		M^{\prime Xj}_{k + \frac{n}{2}} \ket{\psi^{\prime}}
		-
		Z^{\prime}_{k} \ket{\psi^{\prime}}
	}
	\leq
	\sqrt{2 \epsilon}
	.
\end{equation} 
Following an analogous argument, we determine that
\begin{equation}
	\norm{
		M^{\prime Xj}_{\ell + \frac{n}{2}} \ket{\psi^{\prime}}
		-
		X^{\prime}_{\ell} \ket{\psi^{\prime}}
	}
	\leq
	2\sqrt{2 \epsilon}
	.
\end{equation}
It is now straightforward to verify
\begin{eqnarray*}
	\norm{
		X^{\prime}_\ell Z^{\prime}_k \ket{\psi^{\prime}} 
		-
		X^{\prime}_\ell M^{\prime Xj}_{k + \frac{n}{2}} \ket{\psi^{\prime}}
	}
	& \leq & \sqrt{2 \epsilon}
\\
	\norm{
		X^{\prime}_\ell Z^{\prime}_k \ket{\psi^{\prime}}
		-
		M^{\prime Xj}_{k + \frac{n}{2}} X^{\prime}_\ell \ket{\psi^{\prime}}
	}
	& \leq & \sqrt{2 \epsilon}
\\	
	\norm{
		X^{\prime}_\ell Z^{\prime}_k \ket{\psi^{\prime}}
		-
		M^{\prime Xj}_{k + \frac{n}{2}} M^{\prime X_j}_{\ell + \frac{n}{2}} \ket{\psi^{\prime}}
	}
	& \leq & 2\sqrt{2 \epsilon}
\\
	\norm{
		X^{\prime}_\ell Z^{\prime}_k \ket{\psi^{\prime}}
		-
		M^{\prime Xj}_{\ell + \frac{n}{2}} M^{\prime X_j}_{k + \frac{n}{2}} \ket{\psi^{\prime}}
	}
	& \leq & 2\sqrt{2 \epsilon}
\\
	\norm{
		X^{\prime}_\ell Z^{\prime}_k \ket{\psi^{\prime}}
		-
		M^{\prime Xj}_{\ell + \frac{n}{2}} Z^{\prime}_{k} \ket{\psi^{\prime}}
	}
	& \leq & 3\sqrt{2 \epsilon}
\\
	\norm{
		X^{\prime}_\ell Z^{\prime}_k \ket{\psi^{\prime}}
		-
		Z^{\prime}_{k} M^{\prime Xj}_{\ell + \frac{n}{2}} \ket{\psi^{\prime}}
	}
	& \leq & 3\sqrt{2 \epsilon}
\\
	\norm{
		X^{\prime}_\ell Z^{\prime}_k \ket{\psi^{\prime}}
		-
		Z^{\prime}_{k} X^{\prime}_{\ell} \ket{\psi^{\prime}}
	}
	& \leq & 4\sqrt{2 \epsilon}
\end{eqnarray*}
On the second line we have used the fact that $M^{\prime Xj}_{k + \frac{n}{2}}$ is defined on Bob's subsystem, whereas $X^{\prime}_{\ell}$ is on Alice's subsystem.  The sixth line is analogous.  On the fourth line, we have used the fact that $M^{\prime Xj}_{\ell + \frac{n}{2}}$ and $M^{\prime X_j}_{k + \frac{n}{2}}$ commute by construction, and the last line is~\eref{eq:XZapproxcommute} as desired.
\end{proof}

Next we need to show that the $X$'s and $Z$'s anti-commute on the same sub-test.  For this we will appeal to \cite{McKague:2012:Robust-self-tes} Theorem~3, which we summarize here.

\begin{lemma}[McKague et al.\ \cite{McKague:2012:Robust-self-tes}]
\label{lemma:mayersyao}
Given a bipartite state  $\ket{\psi^{\prime}}$ and operators $M^{\prime}_A$ and $N^{\prime}_B$ with $M,N$ ranging over $\{ D,X,Z\}$, if we have for all $M,N$ (excluding the case $M=D=N$)
\begin{eqnarray}
    \left|
    \bra{\psi^{\prime}}
        M^{\prime}_A
        N^{\prime}_B
    \ket{\psi^{\prime}}
    -
    \bra{\psi}
        M_A
        N_B
    \ket{\psi}
    \right|
    \leq \epsilon
\end{eqnarray}
where $\ket{\psi}$ is the graph state of an isolated edge and $A$ and $B$ refer to different subsystems, then
\begin{equation}
\label{eq:XZanticommute}
	\norm{
		Z^{\prime}_A X^{\prime}_A\ket{\psi^{\prime}}
		+ 
		X^{\prime}_A Z^{\prime}_A \ket{\psi^{\prime}}
	}
	\leq 
	4 \left(
        1 + \sqrt{2}
    \right)
    (2 \epsilon)^\frac{1}{4}
    +
    8 \sqrt{2 \epsilon} 
    +
    \left(
        5 + 3 \sqrt{2}
    \right)
    (2 \epsilon)^{\frac{3}{4}}
\end{equation}
and analogously for the $B$ side.
\end{lemma}

We are now ready to prove the main result, that our paralellized Mayers-Yao test is in fact a self-test.

\begin{theorem}
\label{theorem:myparallelselftest}
Given the definitions of operators in section~\ref{sec:generalbehaviour} if for each $q,r \in \{ X, Z\} \cup \{X_j\}_j \cup \{Z_j\}_j$ and $0 < k \leq \frac{n}{2}$ we have
\begin{equation}
\left|
    \bra{\psi^\prime} 
        M^{\prime q}_k
        M^{\prime r}_{k + \frac{n}{2}}
        \ket{\psi^\prime} -
    \bra{\psi} 
        M^{q}_k
        M^{r}_{k + \frac{n}{2}}
    \ket{\psi}
\right|
\leq
\epsilon
\end{equation}
then there exists an isometry $\Phi$ and a state $\ket{junk}$ such that for any $p,q \in (0,1)^{n}$
\begin{eqnarray}
\nonumber
    \norm{
        \Phi(
            X^{\prime q}
            Z^{\prime p}
            \ket{\psi^{\prime}}
        )
        -
        \ket{junk}
        X^q
        Z^p
        \ket{\psi}    
    }
    \leq 
    \sqrt{
        \sqrt{2 \epsilon}
        \left(
            \frac{9 n^2}{4}
            +
            \frac{3 n}{2}
        \right)
        + \frac{n \epsilon_4}{2}
    }    
    +
\\
\nonumber
\qquad
\qquad
    \sqrt{
        \sqrt{2\epsilon}
        \left[
            \frac{9 n^2}{8}        
            +
            n
            \left(
                \frac{5 |p|}{2}
                -
                \frac{1}{4}
            \right)
            -
            \frac{|p|}{2}
        \right]
        +
        \epsilon_4
        \left(
            \frac{n}{4} + 
            \frac{|p|}{2}
        \right)
    } 
    .
\end{eqnarray}
where
\begin{equation}
\label{eq:XZanticommute5}
    \epsilon_4 :=
	4 \left(
        1 + \sqrt{2}
    \right)
    (2 \epsilon)^\frac{1}{4}
    +
    8 \sqrt{2 \epsilon} 
    +
    \left(
        5 + 3 \sqrt{2}
    \right)
    (2 \epsilon)^{\frac{3}{4}}
\end{equation}
\end{theorem}

\begin{proof}
The conditions of the theorem allow us to use Lemma~\ref{lemma:testproperties} to get most of the conditions for Lemma~\ref{lemma:paralleltestsufficientconditions}.  Then for each $k = 1 \dots \frac{n}{2}$ we set $X^{\prime}_A = X^{\prime}_k$, $X^{\prime}_B = X^{\prime}_{k + \frac{n}{2}}$, and similar for $Z$ and $D$.  With these definitions, the conditions of the theorem give us the conditions for Lemma~\ref{lemma:mayersyao}, and we can conclude that equation~\eref{eq:XZanticommute2} in the conditions of Lemma~\ref{lemma:paralleltestsufficientconditions} holds for $k$ and $k + \frac{n}{2}$.  Now we have all the conditions for Lemma~\ref{lemma:paralleltestsufficientconditions}, which gives us $\Phi$ and $\ket{junk}$.  Substituting in for the $\epsilon$'s gets us the final bound.
\end{proof}

Although we have concentrated on a very small number of questions, it is a subset of the questions that would be asked in a strictly parallel version of the Mayers-Yao test.  So a strictly parallel version of Mayers-Yao also self-tests many e-bits.  However it is very inefficient since we ignore almost all of the questions that are asked.

\section{A strictly parallel test}
The parallel Mayers-Yao test developed in the previous section has some interesting properties, but it would be interesting to see how a strictly parallel test can be constructed.  We will do so using a non-local game which incorporates aspects of both the Mayers-Yao and CHSH self-tests which makes it particualarly suited to constructing a self-testing non-local game.

\subsection{Structure of the test and honest behaviour}
First we specify the test and honest behaviour for a single e-bit.  Alice and Bob share the state $\ket{\psi} = \frac{1}{2}(\ket{00} + \ket{01} + \ket{10} - \ket{11}$.

\begin{enumerate}
\item[$A_{X}$] Alice measures her qubit in the $X$ eigenbasis.
\item[$A_{Z}$] Alice measures her qubit in the $Z$ eigenbasis.
\item[$A_{D}$] Alice measures her qubit in the $\frac{X+Z}{\sqrt{2}}$ eigenbasis.
\item[$A_{E}$] Alice measures her qubit in the $\frac{X-Z}{\sqrt{2}}$ eigenbasis.
\end{enumerate}
Bob's behaviour is analogous.

The referee will choose one question for Alice and one for Bob and send them.  However, the referee will never choose the question pairs $A_{X} B_{X}$, $A_{Z} B_{Z}$, $A_{D} B_{D}$, $A_{E} B_{D}$, $A_{D} B_{E}$ or $A_{E} B_{E}$ since we not need them.  There are thus 10 possible question pairs.  Alice and Bob will return their measurement result as either $\pm 1$.  

For the parallel version, testing $\frac{n}{2}$ e-bits, the referee chooses $\frac{n}{2}$ questions for Alice and Bob, one for each e-bit, chosen from the questions for the single copy.  That is to say, the referee generates $\frac{n}{2}$ pairs of questions independently as specified for the single copy test.  There are thus $10^{\frac{n}{2}}$ possible questions.  Honest Alice and Bob measure all of their qubits in the appropriate bases depending on the questions as specified for the single copy and return the results as strings in $\{-1, 1\}^{\frac{n}{2}}$.

Note that there are an exponential number of possible questions, which is necessarily the case for any strictly parallel test.

\subsection{General behaviour}
\label{sec:newtestoperatordefinitions}
Alice and Bob will hold some bipartite state $\ket{\psi^{\prime}}$.  Their measurements can be treated as in section~\ref{sec:generalbehaviour} so that we have operators $M^{\prime q}_k$ for question $q$, sub-test number $k$.  Since this test is strictly parallel, it makes sense to think of the questions as strings of questions.  So, for example, we have question $A_{X \dots X}$ in which Alice is asked to measure $X$ for each sub-test.  We define:
\begin{eqnarray}
    M^{\prime S_1 \dots S_n}_k  & = &
    \cases{
        M^{\prime A_{S_1 \dots S_{\frac{n}{2}} }}_k & $0 < k \leq \frac{n}{2}$\\
        M^{\prime B_{S_{\frac{n}{2} + 1} \dots S_{n}}}_{k - \frac{n}{2}} & $\frac{n}{2} < k \leq n $\\        
    }
\\
    X^{\prime}_k & =  & M^{\prime X \dots X}_k \\
    Z^{\prime}_k & =  & M^{\prime Z \dots Z}_k \\
    D^{\prime}_k & =  & M^{\prime D \dots D}_k \\
    E^{\prime}_k & =  & M^{\prime E \dots E}_k.
\end{eqnarray}

\subsection{Proof of self-testing}
The honest behaviour includes the same correlations as are considered in the CHSH test \cite{Clauser:1969:Proposed-Experi}.  We can thus borrow from previous work \cite{McKague:2012:Robust-self-tes} on CHSH self-testing, where Theorem~2 is summarized here.

\begin{lemma}[McKague et al.\ \cite{McKague:2012:Robust-self-tes}]
\label{lemma:newtestac}
Given a bipartite state  $\ket{\psi^{\prime}} \in \mathcal{H}_1 \otimes \mathcal{H}_2$ and Hermitian and unitary operators $X^{\prime}, Z^{\prime}$ on  $\mathcal{H}_1$ and $D^{\prime}, E^{\prime}$ on $\mathcal{H}_2$ such that
\begin{eqnarray}
    \bra{\psi^{\prime}}
    \left[
        X^{\prime}D^{\prime} -
        X^{\prime}E^{\prime} +
        Z^{\prime}D^{\prime} + 
        Z^{\prime}E^{\prime}
    \right]
    \ket{\psi^{\prime}}
    \geq 2 \sqrt{2} - \epsilon
\end{eqnarray}
where $A$ and $B$ refer to different subsystems, then
\begin{equation}
\label{eq:XZanticommute1}
	\norm{
		Z^{\prime}_A X^{\prime}_A\ket{\psi^{\prime}}
		+ 
		X^{\prime}_A Z^{\prime}_A \ket{\psi^{\prime}}
	}
	\leq 
    2 \sqrt{2\sqrt{2}\epsilon}
\end{equation}
and analogously for the $B$ side.
\end{lemma}

We can also borrow most of Lemma~\ref{lemma:testproperties}.  We must make a small modification in the proof where we use $M^{\prime X_j}_{k + \frac{n}{2}}$ and $M^{\prime X_j}_{\ell + \frac{n}{2}}$, since this operator is not defined here.  Instead we substitute any question $q$ which has $Z$ in the $k$th position and $X$ in the $(k + \frac{n}{2})$th position, and use the operators $M^{\prime q}_{k + \frac{n}{2}}$ and $M^{\prime q}_{\ell + \frac{n}{2}}$.  With this modification we obtain the analogous result to Lemma~\ref{lemma:testproperties}, with the same bounds.  Putting everything together, analogously to Theorem~\ref{theorem:myparallelselftest}, we obtain a self-testing result for this strictly parallel test:

\begin{lemma}
\label{lemma:newtestselftest}
Given the definitions of operators in section~\ref{sec:newtestoperatordefinitions} if for each $k \in \{1 \dots \frac{n}{2}\}$ 
\begin{eqnarray}
\label{eq:newtestconditionschsh1}
    \bra{\psi^{\prime}} \left[    
        X^{\prime}_k
        \left(
            D^{\prime}_{k + \frac{n}{2}}
            - 
            E^{\prime}_{k + \frac{n}{2}}
        \right)
        +
        Z^{\prime}_k
        \left(
            D^{\prime }_{k + \frac{n}{2}}
            +
            E^{\prime}_{k + \frac{n}{2}}
        \right)
        \right] \ket{\psi^{\prime}}
    \geq
    2 \sqrt{2} - \epsilon
\\
\label{eq:newtestconditionschsh2}
    \bra{\psi^{\prime}} \left[    
        X^{\prime}_{k + \frac{n}{2}}
        \left(
            D^{\prime}_k
            - 
            E^{\prime}_{k}
        \right)
        +
        Z^{\prime}_{k + \frac{n}{2}}
        \left(
            D^{\prime }_{k}
            +
            E^{\prime}_{k}
        \right)
        \right] \ket{\psi^{\prime}}
    \geq
    2 \sqrt{2} - \epsilon
\end{eqnarray}
and for every $q,r \in \{X,Y\}^{\frac{n}{2}}$ and $k \in \{1 \dots \frac{n}{2}\}$ such that $r_{k + \frac{n}{2}}$ is the complement of $q_k$ (i.e.,\ if $q_k = X$ then $q_{k + \frac{n}{2}} = Z$ or vice versa)
\begin{equation}
\bra{\psi^{\prime}}
\label{eq:newtestconditionscorrelations}
    M^{\prime q}_k
    M^{\prime r}_{k + \frac{n}{2}}
\ket{\psi^{\prime}}
\geq 1 - \epsilon
\end{equation}
then there exists an isometry $\Phi$ and a state $\ket{junk}$ such that for any $p,q \in (0,1)^{n}$
\begin{eqnarray}
\nonumber
    \norm{
        \Phi(
            X^{\prime q}
            Z^{\prime p}
            \ket{\psi^{\prime}}
        )
        -
        \ket{junk}
        X^q
        Z^p
        \ket{\psi}    
    }
    \leq 
    \sqrt{
        \sqrt{2 \epsilon}
        \left(
            \frac{9 n^2}{4}
            +
            \frac{(3 + 2^{5/4}) n}{2}
        \right)
    }    
    +
\\
\qquad
\qquad
    \sqrt{
        \sqrt{2\epsilon}
        \left[
            \frac{9 n^2}{8}        
            +
            n
            \left(
                \frac{5 |p|}{2}
                -
                \frac{1}{4} +
                \frac{1}{2^{3/4}}                
            \right)
            +
            |p|
            \left(
            2^{1/4}
            -
            \frac{1}{2} 
            \right)            
        \right]
    } 
    .
\end{eqnarray}

\end{lemma}

We do not need all of the conditions in the lemma, specifically we only need enough questions $q$ and $r$ to make the argument in Lemma~\ref{lemma:testproperties} work.  A subset as used in section~\ref{sec:ParallelizingtheMayersYaotest}.

\subsection{A non-local game for the new test}
The Mayers-Yao test, and our parallel extension of it, are not well suited to being phrased as a non-local game where the referee asks a question and decides to accept or reject based on the answers.  The reason is that we have $XX$ and $ZZ$ measurements which should have expected value near zero.  So if the referee observes that the measurement outcome is a 1, should the referee accept or reject?  The CHSH game, on the other hand, is very straightforward to phrase as a non-local game.  In our notation the referee asks Alice either $X$ or $Z$ and Bob either $D$ or $E$.  The referee accepts when Alice's and Bob's answers match, except when the questions are $X$ and $E$, in which case the referee accepts when their answers disagree.  Our new single e-bit test has two CHSH games as sub-tests ($X,Z$ for Alice, and $D,E$ for Bob, and vice versa) each of which must obey the Cirel'son inequality \cite{Cirelson:1980:Quantum-general}, plus two other questions ($X$ for Alice, $Z$ for Bob and vice versa) where the referee accepts when Alice and Bob's answers agree.  It is straightforward to see, then, that the maximum probability of winning the single copy non-local game is
\begin{equation}
\frac{1}{10} \left(
2 + 8 \cos \frac{\pi}{8}
\right)
\end{equation}
or in terms of expectation value (where the referee outputs $1$ for a win and $-1$ for a loss), the maximum is
\begin{equation}
\frac{1}{5}
\left(
2 \sqrt{2}
+ 1
\right).
\end{equation}

Now, moving to the parallel case, the referee picks questions independently for each sub-test.  But what is the winning condition?  The referee has lots of data, and has to boil it down to a single bit: accept or reject.  There are many ways that the referee could do this.   Let us define $A_k$ to be the random variable corresponding to accepting on the $k$-th sub-test, with $1$ meaning accept and $-1$ meaning reject.  The referee will do as follows:
\begin{enumerate}
\item The referee asks questions independently for each of the $\frac{n}{2}$ sub-tests
\item The referee determines $A_k$ for each sub-test -- i.e.,\ whether to accept or reject
\item The referee chooses a number $a$ uniformly at random from $\{-\frac{n}{2} + 1 \dots \frac{n}{2}\}$
\item If $\sum_k A_k \geq a$ the referee accepts, and otherwise rejects
\end{enumerate}

let $A$ be the random variable corresponding to accepting, given the above procedure for the referee, with $A=1$ being accept and $A=-1$ being reject.  Then we have the following lemma:

\begin{lemma}
Let $\{A_k\}_{k=1}^{m}$ be random variables taking values $\pm 1$.  Let $A$ be the random variable defined by the following procedure:
\begin{enumerate}
    \item Observe $A_1 \dots A_m$
    \item Pick $a$ uniformly from $\{-m + 1 \dots m\}$
    \item $A$ is given by
        \begin{equation}
            A =
               \cases{
                1 & $\sum_k A_k \geq a$ \\
                -1 & otherwise
            }
        \end{equation}
\end{enumerate}
Then
\begin{equation}
E(A) = \frac{1}{m} \sum_{k=1}^{m} E(A_k)
\end{equation}
\end{lemma}

\begin{proof}
The expectation is given by:
\begin{equation}
E(A) = \frac{1}{2m}
        \sum_{A_1 \dots A_m} 
            P(A_1 \dots A_m)
        \sum_a
               \cases{
                1 & $\sum_k A_k \geq a$ \\
                -1 & otherwise
            }
\end{equation}
Looking at the inner sum for a fixed $A_1 \dots A_m$, when $a$ takes the values $-m + 1$ through to $\sum_k A_k$ then the summand is 1.  So there are $m + \sum_k A_k$ values of $a$ that will cause the summand to take the value 1, and $m - \sum_k A_k$ that will cause $A$ to take the value $-1$.  Applying this knowledge to the inner sum we find
\begin{equation}
E(A) = \frac{1}{m}
        \sum_{A_1 \dots A_m} 
            P(A_1 \dots A_m)
                \sum_k A_k.
\end{equation}
Rearranging the order of the sum we get
\begin{equation}
E(A) = \frac{1}{m}
        \sum_k
        \sum_{A_1 \dots A_m} 
            P(A_1 \dots A_m)
                 A_k.
\end{equation}
The inner sum is evidently $E(A_k)$, which finishes the proof.

\end{proof}

\begin{theorem}
Given the definition of $A$ from the preceding discussion, if for some $\delta \geq 0$
\begin{equation}
    E(A) 
    \geq
        \frac{1}{5}
        \left(
        2 \sqrt{2}
        + 1
        \right)    
    -
    \delta
\end{equation}
then there exists an isometry $\Phi$ and a state $\ket{junk}$ such that for any $p,q \in (0,1)^{n}$
\begin{eqnarray}
\nonumber
    \norm{
        \Phi(
            X^{\prime q}
            Z^{\prime p}
            \ket{\psi^{\prime}}
        )
        -
        \ket{junk}
        X^q
        Z^p
        \ket{\psi}    
    }
    \leq 
    10^{\frac{n}{8}}
    \sqrt{
        \sqrt{n \delta}
        \left(
            \frac{9 n^2}{4}
            +
            \frac{(3 + 2^{5/4}) n}{2}
        \right)
    }    
    +
\\
\qquad
    10^{\frac{n}{8}}
    \sqrt{
        \sqrt{n \delta}
        \left[
            \frac{9 n^2}{8}        
            +
            n
            \left(
                \frac{5 |p|}{2}
                -
                \frac{1}{4} +
                \frac{1}{2^{3/4}}                
            \right)
            +
            |p|
            \left(
            2^{1/4}
            -
            \frac{1}{2} 
            \right)            
        \right]
    } 
    .
\end{eqnarray}

\end{theorem}
\begin{proof}

Let us introduce some notation to help.  The function $f(q,k)$ gives the winning condition for sub-test $k$ given question $q$.  Note that $f(q,k)$ depends only on the $k$ and $k + \frac{n}{2}$ positions of $q$, i.e.,\ the questions for sub-test $k$.  The expected value for $A_k$ when the question $q$ is asked is then 
$
    f(q,k)
    \bra{\psi^{\prime}} 
        M^{\prime q}_k 
        M^{\prime q}_{k + \frac{n}{2}} 
    \ket{\psi^{\prime}}
$
and
\begin{equation}
E(A) = 
    \frac{2}{10^{\frac{n}{2}}n}
    \sum_{q,k} f(q,k) 
        \bra{\psi^{\prime}} 
            M^{\prime q}_k 
            M^{\prime q}_{k + \frac{n}{2}} 
        \ket{\psi^{\prime}}
\end{equation}
where $q$ ranges over the $10^\frac{n}{2}$ possible questions.  $E(A)$ is bounded above by 
$
\frac{1}{5}
\left(
2 \sqrt{2}
+ 1
\right)
$ since it is the average of the expectations of $k$ values $E(A_k)$, each of which also has this upper bound.

Now let us estimate the value
\begin{equation}
    S := 
    \bra{\psi^{\prime}} \left[    
        X^{\prime}_k
        \left(
            D^{\prime}_{k + \frac{n}{2}}
            - 
            E^{\prime}_{k + \frac{n}{2}}
        \right)
        +
        Z^{\prime}_k
        \left(
            D^{\prime }_{k + \frac{n}{2}}
            +
            E^{\prime}_{k + \frac{n}{2}}
        \right)
        \right] \ket{\psi^{\prime}}.
\end{equation}
Clearly $S \leq 2 \sqrt{2}$ since this is just CHSH correlations.  Note that $S$ is part of the sum making up $E(A)$. Taking the pessimistic view that all other correlations meet their maximum, and that all of the error $\delta$ that we see is due to $S$ being smaller than maximum, we find that $S \geq 2 \sqrt{2} - \frac{2}{10^\frac{n}{2} n} \delta$.  A similar argument applies for the $A$ and $B$ sides swapped.  Letting 
\begin{equation}
\epsilon = \frac{2}{10^\frac{n}{2}n} \delta
\end{equation}
we obtain conditions~\eref{eq:newtestconditionschsh1} and~\eref{eq:newtestconditionschsh2}.

Now for some $q,r$ and $k$ such that $r_{k + \frac{n}{2}}$ is the complement of $q_k$ we look at
\begin{equation}
T :=
\bra{\psi^{\prime}}
    M^{\prime q}_k
    M^{\prime r}_{k + \frac{n}{2}}
\ket{\psi^{\prime}}.
\end{equation}
Straightforwardly $T \leq 1$.  Again supposing that all other correlations meet their maximum, we find that $T \geq 1 - \epsilon$, giving us~\eref{eq:newtestconditionscorrelations}.

We have established the conditions for Lemma~\ref{lemma:newtestselftest}, which gives us the desired conclusion.
\end{proof}

Although we have phrased the non-local game as being strictly parallel, the conditions of Lemma~\ref{lemma:newtestselftest} do not need all of the information that a strictly parallel test gives.  We can easily define a non-local game using a much smaller set of questions as in Section~\ref{sec:ParallelizingtheMayersYaotest}.  Much of the above discussion could be adapted to this set of questions, resulting in a non-local game that tests $m$ e-bits using $O(\log m)$ questions and with polynomial scaling in the robustness.  We leave the exact details for future work.

We may also modify the referee's processing slightly without changing the basic result.  The referee can choose a sub-test $k$ uniformly random and just output $A_k$.  Then $E(A) = \frac{2}{n} \sum_k E(A_k)$, exactly as above.

\section{Discussion}
We have introduced techniques for performing many self-tests in parallel in the case where we do not have no-signalling restrictions between tests.  We gave two constructions which allow for testing many e-bits which are shared between two parties, and we have no other restrictions on the structure of the state or measurements.

Clearly testing in parallel is quite powerful, for example allowing us to test $m$ e-bits using $O(\log m)$ different questions, i.e., $O(\log \log m)$ bits of randomness.  An open question is whether adding more questions (for example, all of the questions in a strictly parallel test) can improve the robustness or not.  For the strictly parallel non-local game presented here, adding more questions reduces the robustness since we ignore most of the questions and they simply mean that the relevant questions get asked with lower probability.

These results open up many possible directions for future work.  Clearly requiring only $O(\log \log m)$ bits of randomness is a remarkable property which could see use in applications.  Improving the robustness scaling in the strictly parallel test is also desirable.  Of course it should also be possible to apply the same techniques to other self-tests, especially CHSH, allowing them to be used in parallel as well.

{\bf Acknowledgements}  This work is funded by the University of Otago and the Dodd-Walls Centre for Photonic and Quantum Technologies.  Thanks to Michael Albert, Valerio Scarani and Carl Miller for valuable feedback.

\section{References}
\bibliographystyle{halphamm}
\bibliography{Global_Bibliography}

\end{document}